\documentclass[sigconf]{acmart}
\usepackage{etex}

\usepackage[font={bf,sf,footnotesize}]{caption}

\usepackage{graphicx}
\usepackage{booktabs}
\usepackage{epstopdf}
\usepackage{url}
\usepackage{placeins}
\usepackage{amsmath,amssymb,amsfonts}
\usepackage{mathtools,mathrsfs}
\usepackage{multirow}
\usepackage{rotating}
\usepackage{pbox}
\usepackage[normalem]{ulem}
\usepackage{listings}

\usepackage{cleveref}
\usepackage[utf8]{inputenc}
\crefname{section}{§}{§§}
\Crefname{section}{§}{§§}

\usepackage[10pt]{moresize}

\usepackage{xcolor}
\definecolor{darkgrey}{RGB}{70,70,70}
\definecolor{lightgrey}{RGB}{200,200,200}

\usepackage{enumitem}

\usepackage{algorithm}

\usepackage{algpseudocode}

\algnewcommand{\LeftComment}[1]{\State \(\triangleright\) #1}

\newcommand{\macb}[1]{\textbf{\textsf{#1}}}

\usepackage{bm}

\usepackage{scalerel,stackengine}
\stackMath
\newcommand\rwh[1]{%
\savestack{\tmpbox}{\stretchto{%
  \scaleto{%
      \scalerel*[\widthof{\ensuremath{#1}}]{\kern-.6pt\bigwedge\kern-.6pt}%
          {\rule[-\textheight/2]{1ex}{\textheight}}
            }{\textheight}%
}{0.5ex}}%
\stackon[1pt]{#1}{\tmpbox}%
}

\def\HiLiGA{\leavevmode\rlap{\hbox to \hsize{\color{black!10}\leaders\hrule height 1\baselineskip depth 1ex\hfill}}}
\def\HiLiGB{\leavevmode\rlap{\hbox to \hsize{\color{black!25}\leaders\hrule height 1\baselineskip depth 1ex\hfill}}}
\def\HiLiGC{\leavevmode\rlap{\hbox to \hsize{\color{black!40}\leaders\hrule height 1\baselineskip depth 1ex\hfill}}}
\def\HiLiGD{\leavevmode\rlap{\hbox to \hsize{\color{black!55}\leaders\hrule height 1\baselineskip depth 1ex\hfill}}}
\def\HiLiGE{\leavevmode\rlap{\hbox to \hsize{\color{black!70}\leaders\hrule height 1\baselineskip depth 1ex\hfill}}}
\def\HiLiGF{\leavevmode\rlap{\hbox to \hsize{\color{black!85}\leaders\hrule height 1\baselineskip depth 1ex\hfill}}}

\usepackage{algpseudocode}
\usepackage{tikz}
\usetikzlibrary{calc}
\usepackage{xcolor}
\makeatletter

\usepackage{listings}

\usepackage{cleveref}
\usepackage[utf8]{inputenc}
\crefname{section}{§}{§§}
\Crefname{section}{§}{§§}

\usepackage{subfigure}
\usepackage{soul}
\usepackage{amsmath}
\usepackage{amsfonts}
\usepackage{url}
\usepackage{setspace}

\usepackage{makecell}

\DeclareMathOperator*{\nnz}{nnz}
\DeclareMathOperator*{\lcm}{lcm}
\DeclareMathOperator*{\flops}{flops}
\usepackage{color}
\definecolor{mygreen}{rgb}{0,0.2,0}
\definecolor{mygray}{rgb}{0.5,0.5,0.5}
\definecolor{mymauve}{rgb}{0.58,0,0.82}
\definecolor{mypurple}{rgb}{0.38,0,0.32}
\definecolor{myblue}{rgb}{0.1,0,0.32}
\newcommand{\costyle}{\fontsize{8.5pt}{9pt}\ttfamily\bfseries}
\newcommand{\kwstyle}{\costyle\textcolor{myblue}}

\newcommand{\CD}{\small\ttfamily\bfseries}
\newcommand{\bhead}[1]{\noindent \textbf{#1:}}
\newcommand{\opbr}{\otimes}

\newcommand{\wA}{A}

\newcommand{\opMM}[2]{\bullet_{\langle #1,#2\rangle}}
\newcommand{\MM}[4]{#3\opMM{#1}{#2}#4}

\algrenewcommand\algorithmicrequire{\textbf{Input}}
\algrenewcommand\algorithmicensure{\textbf{Output}}

\renewcommand{\LeftComment}[1]{\Statex\hskip\ALG@thistlm {$\triangledown$\ {\color{darkgray}#1}}\ $\triangledown$}
 
\algnewcommand{\IIf}[1]{\State\algorithmicif\ #1}
\algnewcommand{\IThen}{\algorithmicthen\ }
\algnewcommand{\EndIIf}{\unskip}
\algnewcommand{\IElse}{\State\algorithmicelse\ \unskip\ }
\algdef{SE}[DOWHILE]{Do}{doWhile}{\algorithmicdo}[1]{\algorithmicwhile\ #1}%

\newcommand*{\affmark}[1][*]{\textsuperscript{#1}}

\begin{document}

\title{Scaling Betweenness Centrality using Communication-Efficient Sparse Matrix Multiplication}

\author{ Edgar Solomonik\affmark[1], Maciej Besta\affmark[2], Flavio Vella\affmark[3], Torsten Hoefler\affmark[2]\\
 {\normalsize\affmark[1]Department of Computer Science, University of Illinois at Urbana-Champaign}
\\ {\normalsize\affmark[2]Department of Computer Science, ETH Zurich}
\\ {\normalsize\affmark[3]Sapienza University of Rome}
\\ {\normalsize solomon2@illinois.edu, maciej.besta@inf.ethz.ch, vella@di.uniroma1.it, htor@inf.ethz.ch} }

\copyrightyear{2017}
\acmYear{2017}
\setcopyright{acmlicensed}
\acmConference{SC17}{November 12--17, 2017}{Denver, CO, USA}\acmPrice{15.00}\acmDOI{10.1145/3126908.3126971}
\acmISBN{978-1-4503-5114-0/17/11}

\begin{abstract}
Betweenness centrality (BC) is a crucial graph problem that measures the
significance of a vertex by the number of shortest paths leading through it.
We propose Maximal Frontier Betweenness Centrality (MFBC): a succinct BC
algorithm based on novel sparse matrix multiplication routines that performs a
factor of $p^{1/3}$ less communication on $p$ processors than the best known
alternatives, for graphs with $n$ vertices and average degree $k=n/p^{2/3}$. We
formulate, implement, and prove the correctness of MFBC for weighted graphs by
leveraging monoids instead of semirings, which enables a surprisingly succinct
formulation.
MFBC scales well for both extremely sparse and relatively dense graphs. It
automatically searches a space of distributed data decompositions and sparse
matrix multiplication algorithms for the most advantageous configuration. The
MFBC implementation outperforms the well-known CombBLAS library by up to 8x and
shows more robust performance.  
Our design methodology is readily extensible to other graph problems.

\end{abstract}

\maketitle

\section{Introduction}
\label{sec:intro}

Graph processing underlies many computational problems in machine learning,
computational science, and other
disciplines~\cite{DBLP:journals/ppl/LumsdaineGHB07}.  Yet, many parallel graph
algorithms struggle to achieve scalability due to irregular communication
patterns, high synchronization costs, and lack of spatial or temporal cache locality. 
To alleviate this, we pursue the methodology of formulating scalable graph
algorithms via sparse linear algebra primitives.

In this paper, we focus on betweenness centrality (BC), an important graph
problem that measures the significance of a vertex $v$ based on the number of
shortest paths leading through $v$. BC is used in analyzing various networks in
biology, transportation, and terrorism
prevention~\cite{bader2007approximating}.  The Brandes BC
algorithm~\cite{prountzos2013betweenness, brandes2001faster} provides a
work-efficient way to obtain centrality scores without needing to store all
shortest-paths simultaneously, achieving a quintessential reduction in the
memory footprint.  To date, most parallelizations of BC have leveraged the
breadth first search (BFS) primitive~\cite{deng2009taming, tan2009parallel,
tan2011analysis, Bader:2006:PAE:1156433.1157604, madduri2009faster,
green2013faster, mclaughlin2014scalable}, which can be used to implement
Brandes algorithm on unweighted graphs~\cite{kepner2011graph}.

Contrarily to these schemes, we propose to use the Bellman-Ford shortest path
algorithm~\cite{bellman1956routing, ford1956network} to compute shortest
distances and multiplicities. This enables us to achieve maximal parallelism by
simultaneously propagating centrality scores from all vertices that have
determined their final score, starting with the leaves of the shortest path
tree. We refer to this approach as the Maximal Frontier Betweenness Centrality
(MFBC) algorithm, because the frontier of vertices whose edges are relaxed
includes \emph{all} vertices that could yield progress.
We prove the correctness of the new scheme for computing shortest distances,
multiplicities, and centrality scores with a succinct path-based argument
applied to factors of partial centrality scores, simplifying 
the Brandes approach~\cite{brandes2001faster}.

Each set of frontier relaxations in MFBC is done via multiplication of a pair
of sparse matrices. These multiplications are defined to perform the desired
relaxation via the use of monoids, monoid actions, and auxiliary functions.
This algebraic formalism enables concise definition as well as implementation
of MFBC via the Cyclops Tensor Framework (CTF)~\cite{2015arXiv151200066S}. CTF is a
distributed-memory library that supports tensor contraction and summation;
operations that generalize the sparse matrix multiplications MFBC requires.
By implementing a robust set of sparse matrix multiplication algorithms that
are provably communication-efficient, our framework allows rapid implementation
of bulk synchronous graph algorithms with regular communication patterns and
low synchronization cost.

Aside from the need for transposition (data-reordering), sparse tensor
contractions are equivalent to sparse matrix multiplication. We present a
communication cost analysis of the sparse matrix multiplications we developed
within CTF, and use it to derive the total MFBC cost. The theoretical
scalability both of the sparse matrix multiplication, as well as of the MFBC
algorithm surpasses the state-of-the-art. In our evaluation, the new algorithm
obtains excellent strong and weak scaling for both synthetic and real-world
power-law graphs. We compare our implementation to
that of the Combinatorial BLAS (CombBLAS) library~\cite{bulucc2011combinatorial}, a
state-of-the-art matrix-based distributed-memory betweenness centrality code.
We demonstrate that the CTF-MFBC code outperforms CombBLAS by up to 8x and measure a reduction
in the runtime communication and synchronization costs.
While CombBLAS is still faster in some cases, our implementation is general to weighted graphs
and attains more consistent performance across different graphs.

%

\section{Notation, Background, Concepts}
\label{sec:bgrnd}
We first introduce the notation and basic concepts. The most important
employed symbols are summarized in Table~\ref{tab:symbols}.

\subsection{Basic Graph Notation}

We start by presenting the used graph notation.  We represent an undirected
unweighted labeled graph $G$ as a tuple $(V,E)$; $V=1:n$ ($V = \{1,\ldots, n\}$) is a
set of vertices and $E \subseteq V \times V$ is a set of edges; $|V|=n$ and
$|E|=m$.  We denote the set of possible weight values as $\mathbb{W}\subset
\mathbb{R}\cup \{\infty\}$.  If $G$ is weighted, we have $G = (V,E,w)$ where
$w: E \to \mathbb{W}$ is a weight function.  We denote the adjacency matrix of
$G$ as $\wA$; $\wA(i,j)=w(i,j)$ if $(i,j)\in E$, otherwise $\wA(i,j)=\infty$. 

\begin{table}
\centering
\footnotesize
\sf
\begin{tabular}{@{}l|ll@{}}
\toprule
\multirow{6}{*}{\begin{turn}{90}\shortstack{{Graph}\\{structure}}\end{turn}} & $G$&a given graph, $G=(V,E)$; $V$ and $E$ are sets of vertices and edges;\\
                   & & if $G$ is weighted, then $G=(V,E,w)$ where $w : E \to \mathbb{W}$.\\
                   & $n,m$&numbers of vertices and edges in $G$; $n = |V|, m = |E|$.\\
                   & $\rho(v)$&degree of $v$; $\overline{\rho}$ and $\rwh{\rho}$ are the average and maximum degree in $G$.\\
                   & $d$& diameter of a given graph $G$.\\ 
                   & $A$& adjacency matrix of $G$.\\
                   \midrule
\multirow{6}{*}{\begin{turn}{90}\shortstack{General BC}\end{turn}} 
                   & $\lambda(v)$ & betweenness centrality of $v$.\\ 
                   & $\tau(s,t)$&shortest path distance between $s,t$.\\
                   & $\bar\sigma(s,t)$& number of shortest paths between $s,t$.\\
                   & $\sigma(s,t,v)$& number of shortest paths between $s,t$ leading via $v$.\\
                   & $\delta(s,v)$ &  dependency of $s$ on $v$; $\delta(s,v) = \sum_{t \in V} {\sigma(s,t,v)}/{\bar\sigma(s,t)}$.\\
                   & $\pi(s,v)$ & set of immediate predecessors of $v$ in shortest paths from $s$ to $v$.\\
\bottomrule
\end{tabular}
\caption{Symbols used in the paper; $v,s,t \in V$ are vertices.}
\label{tab:symbols}
\end{table}

\subsection{Basic Algebraic Structures}

We use heavily two structures: semirings and monoids.

\macb{Monoids}
A monoid $(S, \oplus)$ is a set $S$ closed under an associative binary operation
$\oplus$ with an identity element. A commutative monoid $(S,\oplus)$ is a
monoid where $\oplus$ is commutative, and

\small
\begin{gather*}
\bigoplus_{i=j}^k s(i) = s(j)\oplus s(j+1)\oplus \cdots \oplus s(k)
\end{gather*}
\normalsize
for any $k\geq j$ with $s(i)\in S$ for each $i\in j:k$.  We denote the
elementwise application of a monoid operator to a pair of matrices as
$A\oplus B$ for any $A,B\in S^{m\times n}$.

\macb{Semirings}
%
%
A semiring is defined as a tuple $S = (T, \oplus, \otimes)$; $T$ is a set equipped
with two binary operations $\oplus$, $\otimes$ such that $(T, \oplus)$ is a
commutative monoid and $(T, \otimes)$ is a monoid.
Semirings are often used to develop graph algorithms based on linear algebra
primitives. In this work, we use monoids to enable a succinct formulation.

\subsection{Algebraic Graph Algorithms}

Most graph algorithms can be expressed via matrix-vector or matrix-matrix
products. As an introductory example, we consider
BFS~\cite{Cormen:2001:IA:580470}. BFS starts at a root vertex $r$ and traverses
all nodes connected to $r$ by one edge, then the set of nodes two edges away
from $r$, etc. BFS can be used to compute shortest paths in an unweighted
graph, which we can represent by an adjacency matrix with elements
$A_{ij}\in\{1,\infty\}$. In this case, BFS would visit each vertex $v$ and
derive its distance $\tau(r,v)$ from the root vertex $r$.

Algebraically, BFS can be expressed as iterative multiplication of the sparse
adjacency matrix $A$ with a sparse vector $x_i$ over the tropical semiring,
($i$ denotes the iteration number). The tropical semiring is a commutative
monoid $(\mathbb{W}, \min)$ combined with the addition operator (replacing the
monoid $(\mathbb{R},+)$ and multiplication operator that are usually used
for the matrix-vector product). The BFS algorithm would initialize ${x}_0^r
= 0$ (the initial distance to $r$ is $0$) and any other element of ${x}_0$ is
$\infty$ (i.e., the initial distance to any other element is $\infty$).  Each
BFS iteration computes $x_{i+1}=\MM{\min}{+}{x_i}{A}$, then screens $x_{i+1}$
retaining only elements that were $\infty$ in all $x_j$ for $j<i$.  The
sparsity of the vector is given by all entries which are not equal to $\infty$,
the additive identity of the tropical semiring.

\subsection{Brandes' Algorithm}
\label{sec:bc-spmm}

BC derives the importance of a vertex $v$ using the number of shortest paths
that pass through $v$. Let $\bar\sigma(s,t)$ be the number of shortest paths
between vertices $s,t$, and let $\sigma(s,t,v)$ be the number of such paths
that pass via $v$, $\sigma(s,s,v)=\sigma(s,t,s)=\sigma(s,t,t)=0$.  The
centrality score of $v$ is defined as $\lambda(v) = \sum_{s,t \in V}
\frac{\sigma(s,t,v)}{\bar\sigma(s,t)}$.  Define the dependency of a source
vertex $s$ on $v$ as: $\delta(s,v) = \sum_{t \in V}
\frac{\sigma(s,t,v)}{\bar\sigma(s,t)}$. Then, $\lambda(v) = \sum_{s  \in V}
\delta(s,v)$ where $\delta(s,v)$ satisfies the recurrence: $\delta(s,v) =
\sum_{w\in \pi(s,w)} \frac{\bar\sigma(s,v)}{\bar\sigma(s,w)} (1 +
\delta(s,w))$; $\pi(s,w)$ is a list of immediate predecessors of $w$ in the
shortest paths from $s$ to $w$.  Brandes' scheme~\cite{brandes2001faster} uses
this recurrence to compute $\lambda(v)$ in two parts. First, multiple
concurrent BFS traversals compute $\pi(s,v)$ and $\bar\sigma(s,v)$,
$\forall{s,v \in V}$, obtaining a predecessor tree over $G$. Second, the tree
is traversed backwards (from the highest to the lowest distance) to compute
$\delta(s,v)$ and $\lambda(v)$ based on the equations above.

Brandes' algorithm has been a subject of many previous efficiency studies
~\cite{jin2010novel,yang2005parallel,Bader:2006:PAE:1156433.1157604,madduri2009faster,prountzos2013betweenness,Kulkarni:2007:OPR:1250734.1250759,green2013faster,tan2009parallel,DBLP:journals/corr/abs-0809-1906,cong2012optimizing,tan2011analysis,shi2011fast,yang2011fast,sariyuce2013betweenness}.
Some efforts considered distributed memory
parallelization~\cite{edmonds2010space,Bernaschi2016,fan2017gpu}.  The only distributed
memory BC implementation done using matrix primitives we are aware of exists in
CombBLAS~\cite{bulucc2011combinatorial}.  To the best of our knowledge, we
provide the first communication cost analysis of a BC algorithm, and the first
implementation leveraging 3D sparse matrix multiplication.  A mix of graph
replication and blocking have been previously used for BC
computation~\cite{Bernaschi2016}, but the communication complexity of the
scheme was not analyzed.  Furthermore, previous parallel codes and algebraic BC
formulations have largely been limited to unweighted graphs.

\section{Monoids for Succinct Formulation}
\label{sec:monoids}
%
Our first idea is to employ commutative {monoids} for describing MFBC.
Semirings permit multiplicative operators only on elements within the same set,
while our algorithms require operators on members of different sets. We use
  monoids to define generalized matrix-vector and matrix-matrix multiplication
  operators.

%
%
%

To define a suitable matrix multiplication primitive for our algorithms, we
permit different domains for the matrices and replace elementwise multiplication
with an arbitrary function that is a suitable map between the domains. 
Specifically, consider two input matrices $A\in D_A^{m\times k}$ and $B\in
D_B^{k\times n}$, a bivariate function $f : D_A\times D_B \to D_C$,
and a commutative monoid $(D_C,\oplus)$.
Then, we denote matrix multiplication (MM) as $C=\MM{\oplus}{f}{A}{B}$, where each element 
of the output matrix, $C \in D_C^{m \times n}$, is
$C(i,j)=\bigoplus_{k=1}^nf(A(i,k), B(k,j))$, $\forall i\in 1:m, j\in 1:n$.
This MM notation enables a unified description of the main steps of
MFBC.

\section{Maximal Frontier Algorithm}
\label{sec:bc-new-alg}

We now describe our algebraic maximal frontier BC algorithm (MFBC), which uses
the introduced algebraic formulation based on monoids for a succinct
description.
MFBC consists of two parts.  First, it enhances Bellman-Ford to compute
distances between vertices ($\tau(s,t)$) and the multiplicities of shortest
paths ($\bar\sigma(s,t)$); we refer to this part as Maximal Frontier Bellman
Ford (MFBF; see Section~\ref{sec:mfbf}). Second, it computes partial
centrality factors ($\delta(s,v)$) with a strategy
that extends Brandes' algorithm; we thus refer to it as Maximal Frontier
Brandes (MFBr; see Section~\ref{sec:mfbr}). 

Both MFBF and MFBr use generalized matrix products of: (1) the adjacency matrix
$A$, and (2) a sparse rectangular $n\times n_b$ matrix with elements of one of
the two specially defined monoids: \emph{multpaths} (elements from the
\emph{multpath monoid} associated with MFBF) and \emph{centpaths} (elements
from the \emph{centpath monoid} associated with MFBr).
Here, $n_b$ is the \emph{batch size}: the number of vertices for which we
solve the final BC score $\lambda(v)$. It constitutes a tradeoff between
the time and the storage complexity: MFBC takes $n / n_b$ iterations but 
must maintain an $n \times n_b$ matrix.
Finally, these two combined schemes give MFBC (i.e., informally MFBC = MFBF +
MFBr).

%
%

\subsection{MFBF: Computing Shortest Paths}
\label{sec:mfbf}

The core idea behind MFBF is to use and extend 
Bellman-Ford so that it computes not only shortest paths but also
their multiplicities.
To achieve this, rather than working only with weights, we define MFBF in
terms of multpaths: tuples (that belong to the multpath monoid) which
carry both path weight and multiplicity.

\subsubsection{The Multpath Monoid}

\leavevmode\\
\bhead{Intuition}
To express Bellman-Ford with multiplicities algebraically, MFBF uses \emph{the multpath monoid} $(\mathbb{M},
\oplus)$.  The elements of $\mathbb{M}$ are \emph{multpaths}: tuples that model
a weighted path with a multiplicity.  The $\oplus$ operator acts on any two
multpaths $x$ and $y$ and returns the one with lower weight; if the weights of
$x$ and $y$ are equal, then their multiplicities are summed.

\bhead{Formalism}
A multpath $x=(x.w,x.m)\in \mathbb{M}=\mathbb{W}\times \mathbb{N}$ is a path in
$G$ with a weight $x.w$ and a multiplicity $x.m$.  Then, we have
{
\begin{gather*}
\forall_{x,y \in \mathbb{M}}, \quad 
\small
x\oplus y=
  \begin{cases} 
    x & : x.w<y.w \\
    y & : x.w>y.w \\
    (x.w,x.m+y.m) & : x.w=y.w. 
  \end{cases}
\end{gather*}
\normalsize
}


\subsubsection{The Bellman-Ford Action}
\label{sec:bfact}

\leavevmode\\
\bhead{Intuition}
In each MFBF iteration, we multiply $A$ with a sparse tensor ($\mathcal{T}$)
that constitutes the \emph{multpath frontier}: it contains the multpaths of the
nodes whose multiplicity changed in the previous iteration.
Here, each element-wise operation acts on a multpath (an element from
$\mathcal{T}$) and an edge weight (an element from $A$); we refer to this
operation as the \emph{Bellman-Ford Action}.


\bhead{Formalism}
Our MFBF algorithm (Algorithm~\ref{alg:MFBF}) iteratively updates a matrix $T$
of multpaths via forward traversals from each source vertex.  This is done in
the inner loop (lines~\ref{li:part_1_inner_start}-\ref{li:part_1_inner_end})
where the $\mathcal{T}$ tensor with partial multiplicity scores is updated in
each iteration using Bellman-Ford Action $\opMM{\oplus}{f}$, where the function
$f$ is defined as
\small
\begin{gather*}
f : \mathbb{M}\times \mathbb{W} \to \mathbb{M}, \quad\quad
f(a,w) = (a.w+w,a.m).
\end{gather*}
\normalsize
$f$ is interpreted as an action of the monoid $(\mathbb{W},+)$ on the set $\mathbb{M}$.
This concept generalizes to $n\times k$ matrices, where we have a monoid action 
$\opMM{\oplus}{f}$ with monoid $(\mathbb{W}^{n\times n}, \opMM{\min}{+})$ on the set $\mathbb{M}^{n\times k}$.

\subsubsection{Algorithm and Correctness}

\leavevmode\\
\bhead{Intuition}
We obtain shortest path distances and multiplicities via MFBF (Algorithm~\ref{alg:MFBF}):
a Bellman-Ford variant that relaxes all edges adjacent to vertices whose path
information changed in the previous iteration.
The edge relaxation is done via matrix multiplication of the adjacency matrix 
and a multpath matrix, which appends edges to the existing frontier of 
vertices via function $f$, then uses the multpath operator $\oplus$ to
select the minimum distance new set of paths, along with the number of such
new paths. This partial multiplicity score is subsequently accumulated to $T$ 
if it corresponds to a minimum distance path from the starting vertex.
Note that the multiplicity is set to 1 (line~\ref{li:part_1_start}) even if the
corresponding weight equals $\infty$. Thus, such edges are considered in the
main loop (line~\ref{li:part_1_inner_start}). When a path to $v$ with a finite
distance is found, it replaces such a multiplicity.

\begin{algorithm}[!]\scriptsize\footnotesize
\caption{$[T]=$\ MFBF$(A,\vec{s})$}
\algrenewcommand\algorithmicindent{0.8em}
\footnotesize
\begin{algorithmic}[1]
  \Require $\wA$: $n\times n$ adjacency matrix, $\vec{s}$: list of starting $n_b$ vertices
  \Ensure $T$: multpath matrix of distances and multiplicities from vertices $\vec{s}$
   \LeftComment{{
Existential qualifiers $\forall s\in 1:n_b$ (denoting starting vertices) and $\forall v\in 1:n$ (denoting destination vertices) are implicit.}}
    \State $T(s,v) := (\wA(\vec{s}(s),v),1)$ \Comment{Initialize multpaths.} \label{li:part_1_start}
    \State $\mathcal{T} := T$ \Comment{Initialize multpath frontier}
    \While{$\mathcal{T}(s,v)\neq (\infty,0)$ for some $s,v$} \label{li:abws:for1}\label{li:part_1_inner_start}
      \State $\mathcal{T} := \MM{\oplus}{f}{\mathcal{T}}{\wA}$ \Comment{Explore nodes adjacent to frontier}
      \State $T := T\oplus \mathcal{T}$ \Comment{Accumulate multiplicities}
      \LeftComment{Determine new frontier based on updates to the vertex path information}
      \IIf {$\mathcal{T}(s,v).m {=} 0 \vee \mathcal{T}(s,v).w {>} T(s,v).w$} 
      \IThen $\mathcal{T}(s,v) {:=} (\infty,0)$ \EndIIf
    \EndWhile \label{li:part_1_end} \label{li:part_1_inner_end}
\end{algorithmic}
\label{alg:MFBF}
\end{algorithm}
\normalsize

\bhead{Formalism} We can prove that Algorithm~\ref{alg:MFBF} will output the correct shortest distances and multiplicites.
\begin{lemma}\label{lem:mfbf}
For any adjacency matrix $\wA$ and vertex set $\vec{s}$, $T=$\ MFBF$(\wA, \vec{s})$ 
satisfies $T(s,v)=(\tau(s,v),\bar\sigma(s,v))$.
\end{lemma}
\begin{proof}
Let the maximum number of edges in any shortest path from node $s$ be $d$.
For $j\in 1:d, v\in V\setminus\{s\}$, let each $h_j(s,v)\in \mathbb{M}$ be a multpath
corresponding to the weight and multiplicity of all shortest paths from vertex
$s$ to vertex $v$ containing {\it up to} $j$ edges (if there are no such paths,
$h_j(s,v)=(\infty,0$).  Further, let each
$\hat{h}_j(s,v)\in \mathbb{M}$ be a multpath corresponding to the weight and
multiplicity of all shortest paths from vertex $s$ to vertex $v$ containing
{\it exactly} $j$ edges (if there are no such paths,
$\hat{h}_j(s,v)=(\infty,0)$). Note that $h_{d}(s,v)$ contains the weight and
multiplicity of all shortest paths from vertex $s$ to vertex $v$, since no
shortest path can contain more than $d$ edges, therefore $h_{d}(s,v)=(\tau(s,v),\bar\sigma(s,v))$.
Further, by the definition of $\oplus$, we have $h_j(s,v)=\bigoplus_{q=1}^j \hat{h}_j(s,v)$.

Let $T_j(s,v)$ be the state of $T(s,v)$ 
after the completion of $j-1$
iterations of the loop from line~\ref{li:abws:for1}.  We show by induction on $j$
that $T_j(s,v)=h_{j}(s,v)$ and $\mathcal{T}_j(s,v)=\hat{h}_j(s,v)$, and subsequently
that after $d-1$ loop iterations, $T(s,v)=T_{d}(s,v)=h_{d}(s,v)$.
For $j=1$, no iterations have completed and we have
$T(s,v)=(\wA(s,v),1)$, as desired.  
For the inductive step, we show that given
$T_j(s,v)=h_{j}(s,v)$, one iteration of the loop on line~\ref{li:abws:for1}
yields $T_{j+1}(s,v)=h_{j+1}(v)$.
We note that by definition of $\oplus$, only paths with a minimal weight
$\mathcal{T}_{j}(s,u).w$ contribute to $\mathcal{T}_{j+1}(s,v)$, and
(again by definition of $\oplus$),
\small
\begin{gather*}
\mathcal{T}_{j+1}(s,v).m=\sum_{w \in P} w.m,\ \ \text{where} \\ 
P = \{\mathcal{T}_j(s,u) : \mathcal{T}_j(s,u).w+\wA(u,v) = \mathcal{T}_{j+1}(s,v).w\}, 
\end{gather*}
\normalsize
i.e., $\mathcal{T}_{j+1}(s,v).m$ is the sum of the
multiplicities of all the minimal weight paths from vertex $s$ to
$v$  consisting of $j+1$ edges. Our expression for $P$ is valid,
since each must consist of a minimal weight path of $k$ edges from vertex $s$
to some vertex $u$, which is given by $\mathcal{T}_{j}(s,u)$ and another edge
$(u,v)$ with weight $\wA(u,v)$.
\end{proof}

\subsection{MFBr: Computing Centrality Scores}
\label{sec:mfbr}

Once we have obtained the distances and multiplicities of shortest paths
from a set of starting vertices via MFBF, we can begin computing the partial centrality scores.
We perform this by traversing the shortest-path tree from the leaves to the root.
This time, the maximal frontier is composed of all vertices whose leaves have just reported their centrality scores.

\subsubsection{Centpath Monoid}

\leavevmode\\
\bhead{Intuition} To propagate partial centrality scores, we use centpaths, which store a 
distance, a contribution to the centrality score, and a counter.
Similarly to the multpath monoid, we define a centpath monoid $(\mathbb{C}, \opbr)$ with an operator
that acts on any two centpaths $x$ and $y$, and returns the one with lower weight;
if the weights of $x$ and $y$ are equal, then the partial centrality factors and counter values of the two centpaths are summed.

\bhead{Formalism}
Instead of working with partial centrality scores $\delta(s,v)$ (defined in Section~\ref{sec:bc-spmm})
we work with partial centrality factors (as in~\cite{sariyuce2013betweenness}):
\small
\[\zeta(s,v) = \delta(s,v)/\bar\sigma(s,v) = \sum_{w\in \pi(s,v)}\Big(\frac{1}{\bar\sigma(s,w)}+\zeta(s,w)\Big).\]
\normalsize
Computing $\zeta$ rather than $\delta$ simplifies the algebraic steps done by
the algorithm and leads to a simpler proof of correctness.
Once we have computed $\zeta$, we can construct $\delta$ simply via multiplication
by elements of $\bar\sigma$, which we have already computed via MFBF.

A centpath $x=(x.w,x.p,x.c)\in \mathbb{C} = \mathbb{W}\times \mathbb{R}\times \mathbb{Z}$ 
is a path with a weight $x.w\in \mathbb{W}$, partial centrality score $x.p\in\mathbb{R}$, and a counter $x.c\in\mathbb{Z}$.
Our algorithm will converge to a centpath $x$ for each pair of starting
and destination nodes $s,v$, where the partial dependency factor $x.p=\zeta(s,v)$.
The counter $x.c$ is used to keep track of the number of predecessors who have not propagated a
partial dependency factor up to the node $v$ in a previous iteration. The counter is decremented until reaching zero, at which point the final centrality scores of all predecessors have been integrated into $x.p$ and it is then propagated from $v$ up to the root $s$.

The centpath monoid operator $\opbr$ is defined as
\small
\begin{gather*}
\forall_{x,y \in \mathbb{C}},\quad 
x \opbr  y=
  \begin{cases} 
    x & : x.w>y.w \\
    y & : x.w<y.w \\
    (x.w,x.p+y.p,x.c+y.c) & : x.w=y.w.
  \end{cases}
  \end{gather*}
\normalsize

\begin{algorithm}[t]
\caption{$[Z]=$\ MFBr$(A,T)$}
\algrenewcommand\algorithmicindent{0.8em}
\footnotesize
\begin{algorithmic}[1]
  \Require $\wA$: $n\times n$ adjacency matrix, $T$: matrix of distances and multiplicities
  \Ensure $Z$: centpath matrix of partial centrality factors $\zeta$
   \LeftComment{{
Existential qualifiers  $\forall s\in 1:n_b$ (denoting starting vertices) are implicit and 
$\forall v\in 1:n$ (denoting intermediate vertices).}}
\LeftComment{Initialize centpaths by finding counting predecessors}\label{li:part_2_start}
    \State $\mathcal{Z}(s,v) := (T(s,v).w, 0, 1)$  
    \State $Z := \mathcal{Z}\opbr (\MM{\opbr}{g}{\mathcal{Z}}{{\wA}^\mathsf{T}})$
 \LeftComment{Initialize centpath frontier}
    \IIf {$Z(s,v).c=0$} \IThen $\mathcal{Z}(s,v) := (T(s,v).w, {1}/{T(s,v).m}, -1)$ \IElse $\mathcal{Z}(s,v)=(\infty,0,0)$ \EndIIf
    
    \While{$\mathcal{Z}(s,v)\neq (\infty,0,0)$ for some $s,v$} \label{li:abws:for2} \label{li:part_2_inner_start}
      \State $\mathcal{Z} := \MM{\opbr}{g}{\mathcal{Z}}{{\wA}^\mathsf{T}}$ \Comment{Back-propagate frontier of centralities}
    \LeftComment{Turn off counters for nodes that already appeared in a frontier}
      \IIf {$Z(s,v).c=0$} \IThen $Z(s,v).c=-1$ \EndIIf
      \State $Z := Z \opbr  \mathcal{Z}$ \Comment{Accumulate centralities and increment counters}
      \LeftComment{Determine new frontier based on counters}
      \IIf {$Z(s,v).c=0$} \IThen 
\State \ \   $\mathcal{Z}(s,v) := (T(s,v).w, Z(s,v).p+{1}/{T(s,v).m}, -1)$ \IElse $\mathcal{Z}(s,v)=(\infty,0,0)$ \EndIIf

    \EndWhile \label{li:part_2_inner_end}
\end{algorithmic}
\label{alg:MFBr}
\end{algorithm}
\subsubsection{Brandes Action}

\leavevmode\\
\bhead{Intuition} Our MFBr algorithm (Algorithm~\ref{alg:MFBr}) iteratively updates a matrix $Z$ of centpaths via backward propagation of partial centrality factors from the leaves of the shortest path tree.

\bhead{Formalism} In the inner loop (lines~\ref{li:part_2_inner_start}-\ref{li:part_2_inner_end}),
$Z$ is computed with $\opMM{\opbr}{g}$, where function $g$ is defined
as
\small
\begin{gather*}
g : \mathbb{C} \times \mathbb{W}\to \mathbb{C}, \quad\quad
g(a,w) = (a.w-w, a.p, a.c).
\end{gather*}
\normalsize
$g$ may be interpreted as an action of the monoid $(\mathbb{W},+)$ on set $\mathbb{C}$.

\subsubsection{Algorithm and Correctness}

\leavevmode\\
\bhead{Intuition} For weighted graphs, a single vertex may appear many times in the frontier
as its shortest path information and multiplicity is corrected, unlike in
traversals in BFS or Dijkstra's algorithms, where the total number of 
nonzeros in the matrix multiplication operand $\mathcal{T}$ sums to $(n-1)n_b$ 
over all iterations. For the Brandes step, we can avoid propagating 
unfinalized information as we already know the structure of the shortest
path trees. 

MFBr (Algorithm~\ref{alg:MFBr}) propagates centrality factors optimally via the counter kept by each
centpath, putting vertices in the frontier only when all of their predecessors
have already appeared in previous frontiers. The counter is initialized to the
number of predecessors, is decremented until reaching $0$, added to a frontier 
and set to $-1$ to avoid re-adding the vertex to another frontier. 
This approach is strictly better
than propagating partial centrality scores, which does not contribute to overall progress.
Moreover, this scheme is much faster than using Dijkstra's algorithm
to compute shortest-paths, since it requires the same number of iterations as 
Bellman Ford (Dijkstra's algorithm requires $n-1$ matrix multiplications).

\bhead{Formalism} We can demonstrate correctness of the algorithm by showing that the
counter mechanism serves to correctly define each frontier in the shortest path tree.
\begin{lemma}\label{lem:mfbr}
For any adjacency matrix $\wA$ and a multpath matrix $T$ containing shortest path distances and multiplicities, $Z=$\ MFBr$(\wA, T)$ 
satisfies $Z(s,v).p=\zeta(s,v)$.
\end{lemma}
\begin{proof}
We prove that the partial BC scores are computed correctly after $d-1$ iterations of the loop in line~\ref{li:abws:for2} if all shortest paths from $\vec{s}$ in $G$ consist of at most $d$ edges.
As before, we denote the shortest distance from node $s$ to $v$ as $\tau(s,v)$ and the multiplicity as $\bar\omega(s,v)$.
We define $k_j(s,v)\in \mathbb{\mathcal{Z}}$ as the sum of all minimal distance paths of at most $j-1$ edges from $s$ ending at $u$ that are on the minimal distance path between $1$ and $v$,
\small
\begin{align*}
k_j(s,v) &= \sum_{(u,\star)\in P_j(s,v)} \frac{1}{\bar{\sigma}(s,u)}, \ \ \text{where} \\ 
P_j(s,v) &= \{(u,\vec{w}) |  l\in 1:j-1,\vec{w}\in (1:n)^l, \nonumber\\
  \tau(s,u)&+A(u,w_1)+A(w_1,w_2)+\ldots + A(w_{l},v) = \tau(s,v)\}
\end{align*}
\normalsize
Since $P_{d}(s,v)$ is the set of all shortest paths between $s$ and $v$
that are parts of shortest paths between $s$ and $u$, for each $u$
there are $\sigma(s,u,v)/\bar\sigma(u,v)$ such paths, and therefore,
\small
\begin{gather*}
k_{d}(s,v)= \sum_{(u,\star)\in P_{d}(s,v)} \frac{1}{\bar{\sigma}(s,u)}= \sum_{u=1}^n \frac{\sigma(s,u,v)}{\bar\sigma(s,u)\bar\sigma(u,v)}.
\end{gather*}
\normalsize
We now show that $k_{j}(s,v)$ can be expressed in terms of $k_{j-1}(s,u)$ for all
$u\in P_1(v)$ (the 1-edge shortest-path neighborhood of $v$ from $s$).
We accomplish this by disjointly partitioning $P_j(s,v)$ into $P_1(s,v)$ and $\bigcup_{u\in P_1(s,v)}P_{j-1}(s,u)$, which yields,
{ \small
\begin{align*}
k_{j}(s,v)
&= \sum_{(u,\star)\in P_1(s,v)} \bigg(\frac{1}{\bar\sigma(s,u)}+\sum_{(w,\star)\in P_{j-1}(s,u)} \frac{1}{\bar\sigma(s,w)}\bigg) \nonumber \\
&= \sum_{(u,\star)\in P_1(s,v)} \bigg(\frac{1}{\bar\sigma(s,u)}+k_{j-1}(s,u)\bigg) \nonumber
\end{align*}}
Let $Z_j(s,v)$ be the state of $Z(s,v)$ and $\mathcal{Z}_j(s,v)$ be the state of $\mathcal{Z}(s,v)$ after the completion of $j-1$ iterations of the loop on line~\ref{li:abws:for2}.
We argue by induction on $j$, that for all $j\in 1:d$, $Z_j(s,v).p=k_j(s,v)=k_d(s,v)$ and 
\small
\[\mathcal{Z}_j(s,v)= (\tau(s,v),1/\bar\sigma(s,v)+\zeta(s,v),-1)\]
\normalsize
if and only if the largest number of edges in any shortest path from any node $u$ to $v$, such that $\tau(s,v)=\tau(s,u)+\tau(u,v)$, is $j-1$.
In the base case, $j=1$ and thus $Z_1(s,v).p=k_j(s,v)=k_d(s,v)=\zeta(s,v)=0$ for all vertices $v$ with no predecessors (leaves in the shortest path tree); these vertices are set appropriately in $\mathcal{Z}_1$.

For the inductive step, the update on line~\ref{li:abws:for2} contributes the appropriate factor of $\frac{1}{\bar\sigma(s,u)}+k_{j-1}(u)$ from each predecessor vertex $u$.
Next, each such predecessor vertex $u$ must have been a member of a single frontier by iteration $j$, since the larger number of edges in any shortest path from $u$ to any node $v$ must be no greater than $j-1$.
Therefore, the counter $Z_j(s,v).c=0$, which means $\mathcal{Z}_j(s,v)$ is set appropriately (for subsequent iterations $k>j$, $\mathcal{Z}_{k}(s,v)=(\infty,0,0)$ since we set $Z_j(s,v).c=-1$ at iteration $j+1$).
\end{proof}

\subsection{Combined BC Algorithm}

To obtain a complete algorithm for BC, we combine MFBF and MFBr into MFBC (Algorithm~\ref{alg:MFBC}).
MFBC is parametrized with a batch size $n_b$ and proceeds by computing MFBF and MFBr to obtain partial centrality factors for $n_b$ starting vertices at a time.
These factors are then appropriately scaled by multiplicities ($\bar\sigma(s,v)$) and accumulated into a vector of total centrality scores.
\begin{algorithm}[t]
\caption{$[\lambda]=$\ MFBC$(A)$}
\algrenewcommand\algorithmicindent{0.8em}
\footnotesize
\begin{algorithmic}[1]
  \Require $\wA$: $n\times n$ adjacency matrix, $n_b$: the batch size
  \Ensure $\lambda$: a vector of BC scores
  \State $\forall v\in 1:n,\quad\lambda(v) := 0$ \Comment{Initialize the BC scores}
  \For{$i\in 1:n/n_b$} \label{li:main_loop}
    \State $[T]=$\ MFBF$(A,(i-1)n_b+1:in_b)$
    \State $[Z]=$\ MFBr$(A,T)$
  
    \LeftComment{Accumulate partial centralities: $\delta(s,v) = \zeta(s,v)\cdot \bar\sigma(s,v)$}
    \State $\forall v\in 1:n,\quad \lambda(v) := \lambda(v) + \sum_{s=1}^{n_b} Z(s,v).p\cdot T(s,v).m$ \label{li:part_2_end}
 \EndFor
\end{algorithmic}
\label{alg:MFBC}
\end{algorithm}

\begin{theorem}\label{thm:main}
For any adjacency matrix $\wA$ and $n_b \in 1:n$,  $\lambda=$\ MFBC$(\wA,n_b)$ 
satisfies $\lambda(v)=\sum_{s,t \in V} \frac{\sigma(s,t,v)}{\bar\sigma(s,t)}$.
\end{theorem}

\begin{proof}
We assume $n \mod n_b =0$, if it does not hold then $n\mod n_b$ disconnected
vertices can be added to $G$ without changing $\lambda$.
For each vertex batch, MFBF computes the correct shortest distances and multiplicities $T$ by Lemma~\ref{lem:mfbf}.
For each $T$, MFBr computes the correct partial centrality scores $Z$ by Lemma~\ref{lem:mfbr}.
Thus, at iteration $i$, $T(s,v).m=\bar\sigma((i-1)n_b+s,v)$ and
$Z(s,v).p=\zeta((i-1)n_b+s,v)$.
Next, over all iterations, line~\ref{li:part_2_end} expands to
\small
\begin{align*}
\lambda(v) &= \sum_{s=1}^n Z(s,v).p\cdot T(s,v).m=\sum_{s\in V}\zeta(s,v) \cdot \bar{\sigma}(s,v) \\
&=\sum_{s\in V}\delta(s,v)=\sum_{s\in V}\sum_{t \in V} \frac{\sigma(s,t,v)}{\bar\sigma(s,t)}.
\end{align*}
\normalsize
\end{proof}

\section{Communication Complexity}
\label{sec:parallel}

MFBC leverages all the available parallelism in the problem to
accelerate overall progress.  We now formally study its scalability by bounding
its communication complexity. 
%
%
We first present a cost model (Section~\ref{sec:cost_model}). 
In Section~\ref{sec:par_spmspm} we derive the communication costs for sparse
matrix multiplication, by far the most expensive operation in MFBC. 
We assume sparse matrices with arbitrary dimensions and nonzero count. 
%
%
In several cases, to concretize the model and derive tighter bounds, we use
matrix multiplication (tensor contraction) routines in CTF. However, this does
not limit our analysis as CTF employs a larger space of sparse matrix
multiplication variants than any previous work.
Our results provide a communication bound that is substantially lower than
previous results for sparse matrix multiplication when the number of nonzeros
is imbalanced between matrices. As this scheme is a critical primitive not only in
graph algorithms, but in numerical algorithms such as multigrid,
this theoretical result is of stand-alone
importance.

Finally, in Section~\ref{sec:cost_analysis}, we express the cost of MFBC in terms of
the communication cost of the sparse matrix multiplications it executes. 
Our analysis shows that the latency (synchronization) cost of MFBC may in several cases be
higher than the best-known all-pairs shortest-path
algorithms by a factor proportional to the number of batches $n /
n_b$, which in turn depends on the available memory.
Simultaneously, MFBC can operate with $O(m/p)$ memory per processor, while all-pairs shortest-paths
algorithm require $\Omega(n^2/p)$ memory per processor.
Finally, the communication bandwidth cost of MFBC is identical or better than all the
used comparison targets.

\subsection{Cost Model}
\label{sec:cost_model}

We use a parallel execution model where we count the number of messages and
amount of data communicated by any processor.  We do not keep track of the
number of CPU operations because, for sparse matrix multiplication, all the
considered algorithms have an optimal computation cost, and for BC our
algorithm is work-optimal in the unweighted case. The computation cost in the
weighted case depends on the number of times each vertex appears in a frontier
during the MFBC execution, which depends on the graph connectivity as well as
the edge weights.

We use the $\alpha-\beta$ model~\cite{solomonik2014tradeoffs} where the
latency of sending a message is $\alpha$ and the inverse bandwidth is $\beta$.
We assume that $\alpha \geq  \beta$.
There are $p$ processes and $M$ (number of words) is the size of a local memory
at every processor.
Next, the cost of collective communication routines (scatter, gather, broadcast,
reduction, and allreduction) on $p$ processors in the $\alpha-\beta$
model is $O(\beta\cdot x +\alpha \cdot \log p)$~\cite{DBLP:journals/corr/AzadBBDGSTW15}; 
$x$ is the maximum number of
words that each processor owns at the start or end of the collective.
Furthermore, the cost of a \emph{sparse reduction} where each
processor inputs a sparse array and the resulting array has $x$ nonzeros is
also $O(\beta\cdot x +\alpha \cdot \log p)$.
Finally, we use  $\nnz(X)$ to denote the number of non-zeros in any matrix $X$ and
$\flops(X,Y)$ to denote the number of nonzero products when multiplying sparse
matrices $X$ and $Y$.

\subsection{Parallel Sparse Matrix Multiplication}
\label{sec:par_spmspm}

We first analyze the product of sparse matrices $A^{m \times k}$
and $B^{k \times n}$ that produce matrix $C^{m \times n}$.
We use algorithms based on 1-, 2-, and 3-dimensional matrix
decomposition.
They all have a computation cost of $O(\flops(A,B)/p)$, 
which we omit.

All the considered algorithms and implementations use matrix blocks that
correspond to the cross product of a subset of columns and a subset of rows of
the matrices. The blocks are chosen obliviously of the matrix structure.
For sparse matrices with a sufficiently large number of nonzeros, randomizing
the row and column order implies that the number of nonzeros of each such block
is proportional to the block size. Thus, we assume that the number
of nonzeros in any block of dimensions $b_1\times b_2$ of a sparse matrix $A^{m
\times k}$ has $O(\nnz(A)b_1b_2/(mk))$ nonzeros so long as $mk/(b_1b_2)\leq p$.
When $\nnz(A)\gg p$, 
the assumption holds true with high probability
based on a balls-into-bins argument~\cite{2015arXiv151200066S}.

We also assume that multiplying any two blocks of equal
dimensions yields about the same number of nonzero products and output
nonzeros. This assumption does not impact our communication-cost
analysis, but allows us to assert that the computational work is asymptotically load-balanced.
Thus, when multiplying blocks of size $b_1\times b_2$ of matrix $A$
and $b_2\times b_3$ of matrix $B$, the number of nonzero operations is
$O(\flops(A,B)b_1b_2b_3/(mnk))$ and the number of nonzeros in the output block
contribution to matrix $C$ is 
\small
\[O\bigg(\min\Big[\frac{\nnz(C)}{mn}b_1b_3,\frac{\flops(A,B)}{mnk}b_1b_2b_3\Big]\bigg).\]  
\normalsize
For a sparse matrix corresponding to uniform
random graphs, the respective numbers are 
\small
\[\flops(A,B)\approx \frac{\nnz(A)}{mk}\frac{\nnz(B)}{kn}mnk=\nnz(A)\nnz(B)/k\]
\normalsize
and $\nnz(C)\approx \min(mn,\flops(A,B))$.
We can invoke the same balls-into-bins argument~\cite{2015arXiv151200066S} 
to argue that the work load-balance assumption holds, by arguing about the induced layout of the $\flops(A,B)$
operations as a third-order cyclically distributed tensor.

\subsubsection{1D Algorithms}

1D decompositions 
are the simplest way to parallelize a matrix multiplication.
There are three variants, each of which replicates one of the matrices and
blocks the others into $p$ pieces.  Variant~A replicates $A$ via
broadcast and assigns each processor a set of columns of $B$ and $C$.  Variant~B
broadcasts $B$ and assigns each processor a set of rows of $A$ and $C$.
Variant~C assigns each processor a set of columns of $A$ and rows of $B$,
computes their product, and uses a reduce to obtain $C$.
The communication cost of version~X of a 1D algorithm for $\text{X}\in\{\text{A},\text{B},\text{C}\}$ is
\small
\[W_\text{X}(X,p)=O(\alpha\cdot \log p+\beta\cdot \nnz(X)).\]
\normalsize

\subsubsection{2D Algorithms}

2D algorithms~\cite{Cannon:1969,Geijn:SUMMA:1997, matmul3d}
block all matrices on a grid of $p_r\times p_c$
processors and move the data in steps to multiply matrices.
2D algorithms can be based on point-to-point or collective
communication. The former are up to $O(\log p)$ faster in latency, but the
latter generalize easier to rectangular processor grids.
The algorithms are naturally extended to handle sparse matrices by treating the matrix blocks as sparse~\cite{DBLP:journals/corr/abs-1109-3739,Ballard:2013:COP:2486159.2486196}.
One of the simplest 2D algorithms is Cannon's algorithm~\cite{Cannon:1969}, which shifts blocks of $A$ and $B$ on a square processor grid, achieving a communication cost of
\small
\[O\bigg(\alpha\cdot \sqrt{p} + \beta \cdot \frac{\nnz(A) + \nnz(B)}{\sqrt{p}}\bigg).\] 
\normalsize
The algorithm is optimal for square matrices, but other variants achieve lower communication cost when the number of nonzeros in the two operand matrices are different.

Our implementation uses three variants of 2D algorithms using broadcasts
and (sparse) reductions.
The variant AB broadcasts blocks of $A$ and $B$ along processor grid rows and
columns, while the variants AC and BC reduce C and broadcast A and
B, respectively.  CTF uses $\lcm(p_r,p_c)$ (lcm is the least common multiple) broadcasts/reductions and adjusts
$p_r$ and $p_c$ so that $\lcm(p_r,p_c)\approx \max(p_r,p_c)$ steps of
collective communication are performed. When each matrix block is sparse with the specified
load balance assumptions, the costs achieved by these 2D algorithms are given
in general by $W_\text{YZ}$ for variants
$\text{YZ}\in\{\text{AB},\text{AC},\text{BC}\}$ as
$W_\text{YZ}(Y,Z,p_r,p_c) = $
\small
\begin{gather*}
O\bigg(\alpha\cdot \max(p_r,p_c)\log(p)+\beta\cdot\Big(\frac{\nnz(Y)}{p_r} + \frac{\nnz(Z)}{p_c}\Big)\bigg)\nonumber
\end{gather*}
\normalsize

\subsubsection{3D Algorithms}

While 2D algorithms are natural from a matrix-distribution perspective, the
dimensionality of the computation suggests the use of 3D decompositions~\cite{dekel:657,matmul3d,snirmatmul,berntsen1989communication,Johnsson:1993:MCT:176639.176642,mccol_tiskin_99}, where
each processor computes a subvolume of the $mnk$ dense products.
3D algorithms have been adapted to sparse matrices, in particular by the Split-3D-SpGEMM 
scheme~\cite{DBLP:journals/corr/AzadBBDGSTW15} that costs 
\small
\begin{gather*}
O\left(\alpha\cdot \sqrt{cp}\log{p} + \beta \cdot \left(\frac{\nnz(A)+\nnz(B)}{\sqrt{cp}} + \frac{\flops(A,B)}{p}\right)\right)
\end{gather*}
\normalsize
by using a the grid of processes that is $\sqrt{p/c} \times \sqrt{p/c} \times c$.

We derive 3D algorithms (and implement in CTF) by nesting the three 1D algorithm variants
with the three 2D algorithm variants. The cost of
the resulting nine 3D variants on a $p_1\times p_2\times p_3$ processor grid with the 1D
algorithm applied over the first dimension is
\small
\begin{align*}
W_\text{X,YZ}&(X,Y,Z,p_1,p_2,p_3)=W_\text{X}(X[p_2,p_3])\\
&+
\begin{cases} 
W_\text{YZ}(Y,Z[p_1],p_2,p_3) & : X=Y, \\
W_\text{YZ}(Y[p_1],Z,p_2,p_3) & : X=Z,\\
W_\text{YZ}(Y[p_1],Z[p_1],p_2,p_3) & : X\notin \{Y,Z\}, \\
\end{cases}
\end{align*}
\normalsize
for $(\text{X},\text{YZ}) \in \{\text{A},\text{B},\text{C}\}\times \{\text{AB},\text{AC},\text{BC}\}$, with notation $X[p_2,p_3]$ denoting that the 1D algorithm operates on blocks of $X$ given from a $p_2\times p_3$ distribution, while $Y[p_1]$ and $Z[p_1]$ refer to 1D distributions.
This cost simplifies to
\small
\begin{align*}
&W_\text{X,YZ}(X,Y,Z,p_1,p_2,p_3)=\\
&O\bigg(\alpha\cdot \max(p_1,p_2)\log(\min(p_1,p_2))+\beta\cdot \frac{\nnz(X)}{p_2p_3}\bigg)\\
&+
\begin{cases} 
O\bigg(\beta\cdot(\frac{\nnz(X)}{p_2}+\frac{\nnz(Z)}{p_1p_3})\bigg) & : X=Y, \\
O\bigg(\beta\cdot(\frac{\nnz(Y)}{p_1p_2}+\frac{\nnz(X)}{p_3})\bigg)  & : X=Z,\\
O\bigg(\beta\cdot(\frac{\nnz(Y)}{p_1p_2}+\frac{\nnz(Z)}{p_2p_3})\bigg)  & :  X\notin \{Y,Z\}. \\
\end{cases}
\end{align*}
\normalsize
The amount of memory used by this algorithm is 
\small
\[M_\text{X,YZ}(X,Y,Z,p,p_1)=O\bigg(\frac{\nnz(X)p_1}{p}+\frac{\nnz(Y)+\nnz(Z)}{p}\bigg).\]
\normalsize
As we additionally consider pure 1D and 2D algorithms, then pick the 1D, 2D, or 3D
variant of least cost.
Provided unlimited memory, the execution time of our sparse matrix multiplication scheme is no greater than
\small
\begin{align*}
&W_\text{MM}(A,B,C,p)=O\bigg(\min_{\substack{p_1,p_2,p_3\in \mathbb{N} \\ p_1p_2p_3=p}} \bigg[\alpha\cdot \max(p_1,p_2,p_3)\log p \\
&+\beta\cdot \Big(\frac{\nnz(A)}{p_1p_2}\delta(p_3) 
+ \frac{\nnz(B)}{p_2p_3}\delta(p_1) + \frac{\nnz(C)}{p_1p_3}\delta(p_2)\Big)\bigg]\bigg),
\end{align*}
\normalsize
where $\delta(x)=1$ when $x\geq 1$ and $\delta(1)=0$.

\subsection{Parallel Betweenness Centrality}
\label{sec:cost_analysis}

We finally use the matrix multiplication analysis to ascertain a communication
cost bound for MFBC, which performs the bulk of the computation
via generalized matrix multiplication.
We focus on unweighted graphs; the associated proofs heavily
use the fact that each vertex appears in a unique frontier.
We then (Section~\ref{sec:weighted-anal}) discuss weighted graphs; in this case
we cannot anymore use the same technique to ascertain bounds on the size of
each frontier.

\begin{theorem}\label{thm:mfbc_cost}
For any unweighted $n$-node $m$-edge graph $G$ with adjacency matrix $A$ and diameter $d$,
on a machine with $p$ processors with $M=\Omega(cm/p)$ words of memory each,
for any $c\in [1,p]$, MFBC (Algorithm~\ref{alg:MFBC}) can execute with
communication cost

\small
\begin{align*}
W_\text{MFBC}(n,m,p,c)
=O\bigg(\alpha\cdot \frac{dn^2}{m}\sqrt{p/c^3} \log p 
+\beta\cdot \Big(\frac{n^2}{\sqrt{cp}}+\frac{n\sqrt{m}}{p^{2/3}}\Big)\bigg).
\end{align*}
\normalsize
\end{theorem}

\begin{proof}
MFBC is dominated in computation and communication cost
by multiplications of sparse matrices, triggered by the operator $\opMM{\oplus}{f}$
within MFBF and $\opMM{\otimes}{g}$ within MFBr.
There are up to $2d+1$ such matrix multiplications in total. 
Without loss of generality, we consider only the $d$ within MFBF, letting $F_{i}$ be
the frontier ($\mathcal{T}$) at iteration $i$ and $G_{i}$ be the output of $\MM{\oplus}{f}{\mathcal{T}}{\wA}$, which include $F_{i+1}$ but can be much denser.
We can then bound the cost of MFBC as
\small
\begin{align*}
&W_\text{MFBC} = O\bigg(\min_{n_b\in 1:n} \Big[ \frac{n}{n_b} \sum_{i=1}^d W_\text{MM}(\wA,F_{i},G_{i},p)\Big] \bigg) \\
&= O\bigg(\min_{n_b\in 1:n} \frac{n}{n_b}\sum_{i=1}^d \min_{\substack{p_1,p_2,p_3\in \mathbb{N} \\ p_1p_2p_3=p}} \bigg[\alpha\cdot \max(p_1,p_2,p_3)\log p \\
&+\beta\cdot \Big(\frac{m}{p_1p_2}\delta(p_3) 
+ \frac{\nnz(F_i)}{p_2p_3}\delta(p_1) + \frac{\nnz(G_i)}{p_1p_3}\delta(p_2)\Big)\bigg]\bigg).
\end{align*}
\normalsize
The MFBC algorithm requires $O(nn_b/p)$ memory to store $T$, therefore, we have $nn_b/p=O(M)$, and select $n_b=cm/n$,
\small
\begin{align*}
W&_\text{MFBC}=
 O\bigg(\frac{n^2}{cm} \sum_{i=1}^d \min_{\substack{p_1,p_2,p_3\in \mathbb{N} \\ p_1p_2p_3=p}} \bigg[\alpha\cdot \max(p_1,p_2,p_3)\log p \\
&+\beta\cdot \Big(\frac{m}{p_1p_2}\delta(p_3) 
+ \frac{\nnz(F_i)}{p_2p_3}\delta(p_1) + \frac{\nnz(G_i)}{p_1p_3}\delta(p_2)\Big)\bigg]\bigg).
\end{align*}
\normalsize
We use a 3D algorithm with $p_1= p_2=\sqrt{p/c}$, $p_3=c$, which replicates
$\wA$ via a 1D algorithm, then employs the BC variant of a 2D algorithm, 
using $O(cm/p)$ memory.
$A$'s replication can be amortized over (up to $d$) sparse matrix multiplications and over the $\frac{n^2}{cm}$ batches, since $A$ is always the same adjacency matrix. Thus,
\small
\begin{alignat*}{2}
W_\text{MFBC}&=&&
 O\bigg(\beta\cdot \frac{cm}{p}+\frac{n^2}{cm} \Bigg(\Big( \sum_{i=1}^d  \bigg[\alpha\cdot \sqrt{p/c}\log p \\
& &&+\beta\cdot \Big( \frac{\nnz(F_i)}{\sqrt{pc}} + \frac{\nnz(G_i)}{\sqrt{pc}}\Big)\bigg]\bigg)\Bigg), \\
\end{alignat*}
\normalsize
and furthermore, over all $n_b$ batches the total cost is
\small
\begin{alignat*}{2}
W_\text{MFBC}=
O\bigg(\alpha \frac{dn^2}{m}\sqrt{\frac{p}{c^3}}\log p 
+ \beta   \bigg[ \frac{cm}{p}+\sum_{i=1}^d\frac{n^2(\nnz(F_i)+\nnz(G_i))}{m\sqrt{pc^3}}\bigg]\bigg).
\end{alignat*}
%
\normalsize
Now, since the graph is unweighted, we know that each vertex appears in a unique frontier, so
$\sum_{i-1}^d\nnz(F_i)\leq nn_b=cm$.
Therefore, each node can be reached from 3 frontiers (the one it is a part of, the previous one, and the subsequent one), therefore
$\sum_{i-1}^d\nnz(G_i)\leq 3cm$.
Then, the total bandwidth cost over all $d$ iterations and $cm/n$ batches is $O(\beta \cdot (n^2/\sqrt{cp} + \frac{cm}p))$.
This cost is minimized for $c=p^{1/3}n^2/m$, so with $M=\Omega(n^2/p^{2/3})$ memory, the cost $O(\beta\cdot n\sqrt{m}/p^{2/3})$ can be achieved.
\end{proof}

\subsubsection{Discussion on Weighted Graphs}
\label{sec:weighted-anal}

Our communication cost analysis can be extended to weighted graphs, provided bounds on $\sum_i \nnz(F_i)$ and $\sum_i \nnz(G_i)$
for each batch. The quantity $\sum_i \nnz(F_i)$ can be bounded given an amplification factor bounding
the number of Bellman-Ford iterations in which the shortest path distance between any given pair of source and destination vertices is changed.
However, we do not see a clear way to bound $\sum_i \nnz(G_i)$ for weighted graphs.
We evaluate MFBC for weighted graphs in the subsequent section, observing a slowdown proportional to the factor of increase 
in the number of iterations with respect to the unweighted case (in the unweighted case it is the diameter $d$). 

\subsubsection{Comparison to Other Analyses}

\sloppy
We are not aware of other communication cost studies of BC, but 
we can compare our approach to those computing the full distance matrix via 
all-pairs shortest-paths (APSP) algorithms, requiring at least $n^2/p$ memory, regardless of $m$.
The best-known APSP algorithms leverage 3D matrix multiplication to obtain
a bandwidth cost of $O(\beta\cdot n^2/\sqrt{cp})$ using $O(cn^2/p)$ memory for any $c\in [1,p^{1/3}]$~\cite{tiskin_apsp}.
MFBC matches this bandwidth cost, while using only $O(cm/p)$ memory.
Further, given sufficient memory $M=\Omega(n^2/p^{2/3})$, our algorithm is up to $\min(n/\sqrt{m},p^{2/3})$ faster.
When also considering an algorithm that replicates the graph as an alternative, the best speed-up achievable by MFBC is for $M=\Theta(n^2/p^{2/3})$ memory with $n/\sqrt{m}=p^{1/3}$, and when $\beta \gg \alpha$, in which case $W_\text{MFBC}(n,n^2/p^{2/3},p,p^{2/3})=O(\beta\cdot n^2/p)$ is $p^{1/3}$ times faster than Floyd-Warshall, path doubling, or Dijkstra with a replicated graph.

\subsubsection{Discussion on Latency}

The Floyd-Warshall APSP algorithm has latency cost $O(\alpha\cdot \sqrt{cp})$,
but a path-doubling scheme can achieve $O(\alpha\cdot \log p)$~\cite{tiskin_apsp} using $O(n^2/p^{2/3})$ memory.
Given this amount of memory, MFBC can achieve a latency cost of 
\small
\[O\bigg(\alpha\cdot d \log p\Big(\frac{n^2}{\sqrt{p}m} +\sqrt{n/\sqrt{m}}\Big)\bigg).\]
\normalsize
It might be possible to improve this latency cost by using different sparse matrix multiplication algorithms.

\subsubsection{Discussion on Scalability}

The capability of our algorithm to employ large replication factors $c$ gives it good strong scalability properties.
If each processor has $M=O(m/p_0)$ memory, it is possible to achieve perfect strong scalability in bandwidth cost using up to $p^{3/2}_0n^3/m^{3/2}$ processors,
while for up to $p^{3/2}_0n^2/m$, 
\small
\[W_\text{MFBC}(n,m,cp_0,c) = \frac{1}{c}W_\text{MFBC}(n,m,p_0,1)\] 
\normalsize
is satisfied, so strong scalability is achieved in all costs from $p_0$ to $p_0^{3/2}n^2/m$ processors.
This range in strong scalability is better than that achieved by the best known square dense matrix multiplication algorithms, $p_0$ to $p_0^{3/2}$~\cite{SD_EUROPAR_2011}.

\section{Implementation}
\label{sec:impl}

We implement two parallel versions of MFBC using CTF. 
The first, CTF-MFBC, uses CTF to dynamically select data layouts without guidance from the developer.
The second, CA-MFBC, predefines the 3D processor grid layout that we used to minimize theoretical communication cost in the proof of Theorem~\ref{thm:mfbc_cost}.
We first summarize the functionality of CTF and explain how it provides the sparse matrix operations necessary for MFBC.
We then give more details on how CTF handles data distribution and communication.
\subsection{From Algebra to Code}
\label{sec:ctf}

CTF permits definition of all well-known algebraic structures
and implements tensor contractions with user-defined addition and
multiplication operators~\cite{2015arXiv151200066S}.
Matrices can encode graphs and subgraphs (frontiers);
tensors of order higher than two can represent hypergraphs.
As graphs are sufficient for the purposes of this paper, we refer only
to CTF matrix operations.
An $n\times n$ CTF matrix is distributed across a 
{\kwstyle{World}} (an MPI communicator), and has attributes for symmetry, sparsity, and the algebraic structure of its elements.
We work with adjacency matrices with weights in a set {\CD
W} 
\begin{lstlisting}[frame=none]
Matrix<W> A(n,n,SP,D,Y);
\end{lstlisting}
where {\CD D} is a {\kwstyle{World}} and {\CD Y} defines the {\kwstyle{Monoid}{\CD<W>}} of weights with minimum as the operator.

CTF permits operations on one, two, or three matrices at a time,
each of which is executed bulk synchronously.
To define an operation, the user assigns a pair of indices (character labels) to each matrix 
(generally, an index for each mode of the tensor).
An example function inverting all elements of a matrix {\CD A} and storing them in {\CD B} is expressed as 
\begin{lstlisting}[frame=none]
Function<int,float>([](int x){ return 1./x; })
B["ij"] = f(A["ij"]);
\end{lstlisting}
All CTF operations may be interpreted as nested loops, where one operation is performed on elements of multidimensional arrays in the innermost loop.
For instance, in terms of loops on arrays {\CD A} and {\CD B}, the above example is
\begin{lstlisting}[frame=none]
for (int i=0; i<n; i++)
  for (int j=0; j<n; j++) {B[i,j] = 1./A[i,j];}
\end{lstlisting}
For contractions, we can define functions with two operands.

We express $\opMM{\oplus}{f}$ from Section~\ref{sec:bfact} by defining functions {\CD u} for operation $\oplus$ and {\CD f} for $f$, and then a {\kwstyle{Kernel}} corresponding to $\opMM{\oplus}{f}$
\begin{lstlisting}[frame=none]
Kernel<W,M,M,u,f> BF;
Z["ij"]=BF(A["ik"],Z["kj"]);
\end{lstlisting}
If {\CD Z} is a matrix with each element in {\CD M} and {\CD A} is the adjacency matrix with elements in {\CD W}, the above CTF operation executes $Z=\MM{\oplus}{f}{A}{Z}$.
One could supply the algebraic structure in a {\kwstyle{Monoid}} when defining the matrix, then use {\kwstyle{Function}} in place of a {\kwstyle{Kernel}}.
However, the latter construct parses the needed user-defined functions as template arguments rather than function arguments, enabling generation of more efficient (sparse) matrix multiplication kernels for blocks.
Having these alternatives enables the user to specify which kernels are intensive and should be optimized thoroughly at compile time, while avoiding unnecessary additional template instantiations.

Other CTF constructs employed by our MFBC code are
\begin{itemize}
\item {\kwstyle{Tensor::write()}} to input graphs bulk synchronously,
\item {\kwstyle{Tensor::slice()}} to extract subgraphs,
\item {\kwstyle{Tensor::sparsify()}} to filter the next frontier,
\item {\kwstyle{Transform}} to modify matrix elements with a function.
\end{itemize}
More information on the scope of operations provided by CTF is detailed by Solomonik and Hoefler~\cite{2015arXiv151200066S}.

\subsection{Data Distribution Management}

CTF enables the user to work obliviously of the data distribution of matrices.
Each created matrix is distributed over all processors using a processor grid
that makes the block dimensions owned by each processor as close to a square as possible.
For each operation (e.g., sparse matrix multiplication),  
CTF seeks an optimal processor grid, considering the space
of algorithms described in Section~\ref{sec:par_spmspm} as well as overheads, such as 
redistributing the matrices.

Transitioning between processor grids and other data distributions are achieved using
three kernels: (1) block-to-block redistribution, (2) dense-to-dense redistribution, (3) sparse-to-sparse redistribution.
Kernel (1) is used for reassigning blocks of a dense matrix to processors on a new grid,
(2) is used for redistributing dense matrices between any pair of distributions, and (3)
is used for reshuffling sparse matrices and data input.
After redistribution, the matrix/tensor data is transformed to a format suitable for summation, multiplication, or contraction.
For dense matrices, this involves only a transposition, 
but for sparse matrices, CTF additionally converts data stored as index--value pairs (coordinate format) to a compressed-sparse-row (CSR)
matrix format.

CTF uses BLAS~\cite{lawson1979basic} routines for blockwise operations whenever possible (for the datatypes and algebraic operations provided by BLAS).
We additionally use the Intel MKL library 
to multiply sparse matrices, with 
three variants: one sparse operand, two sparse operands, and two sparse operands with a sparse output.
Substitutes for these routines are provided in case MKL is not available.
Further, for special algebraic structures or mixed-type contractions, general unoptimized 
blockwise multiplication and summation routines are used.
Users can also provide manually-optimized routines for blockwise operations, 
which we do not leverage for MFBC.
%

CTF predicts the cost of communication routines, redistributions, and blockwise operations based
on linear cost models.
Besides latency $\alpha$ and bandwidth $\beta$, CTF also considers the
memory bandwidth cost and computation cost of redistribution and blockwise operations.
The dimensions of the submatrices on which all kernels are executed for a given mapping can be
derived at low cost a priori.
To determine sparsity of blocks, we scale by either the nonzero fraction of the operand matrix or the estimated nonzero fraction of the output matrix.
Automatic model tuning allows the cost expressions of different kernels to be comparable 
on any given architecture.
CTF employs a model tuner that executes a wide set of benchmarks on a range of processors, designed to make use of all kernels for various input sizes.
Tuning is done once per architecture or whenever a kernel is added or significantly modified.

%

\section{Evaluation}
\label{sec:exp}

\begin{figure*}[t]
\vspace{-1.5em}
\centering
\subfigure[MFBC on largest SNAP graphs]{
\includegraphics[width=2.3in]{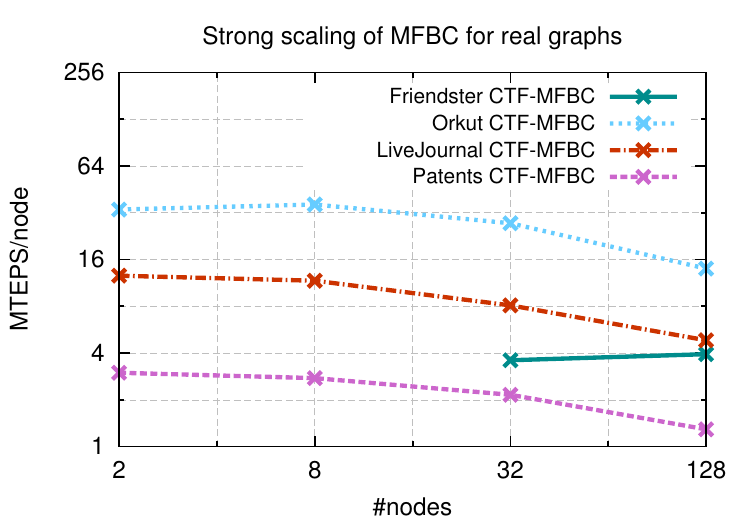}
\label{fig:ctf_snap}
}%
\subfigure[CombBLAS on largest SNAP graphs]{
\includegraphics[width=2.3in]{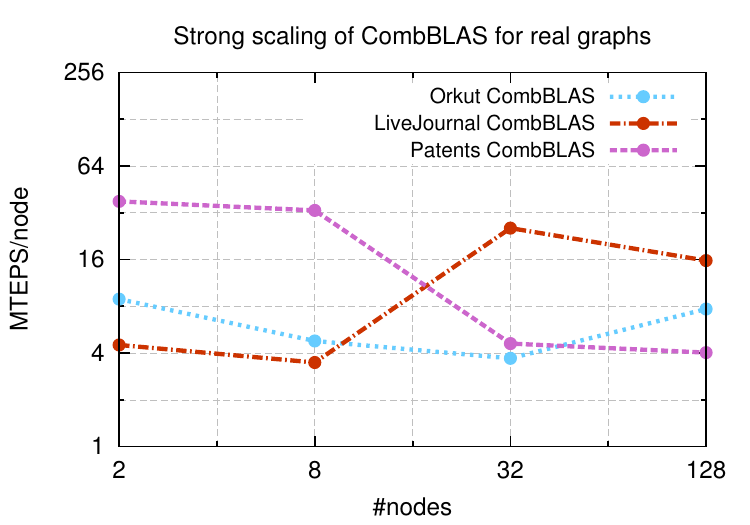}
\label{fig:combblas_snap}
}%
\subfigure[Weighted and unweighted R-MAT graphs]{
\includegraphics[width=2.3in]{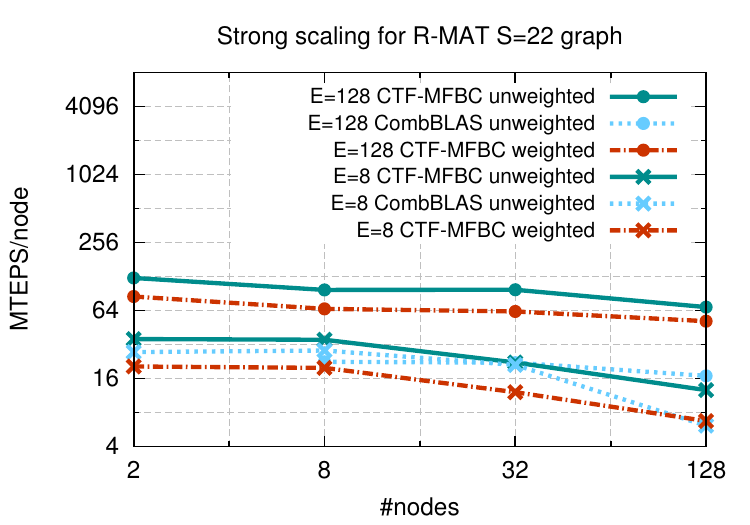}
\label{fig:combblas_ctf_rmat}
}
\vspace{-1em}
\caption{Strong scaling of MFBC and CombBLAS real graphs (Table~\ref{tab:graphs}) and for R-MAT graphs, which have roughly $2^S$ vertices and average degree $E$.
Weights are selected randomly between 1 and 100 for weighted R-MAT graphs in Figure~\ref{fig:combblas_ctf_rmat}}
\label{fig:ss}
\vspace{-1em}
\end{figure*}

We present performance results for the the CTF (unmodified v1.4.2) implementation of MFBC (CTF-MFBC) and for CombBLAS.
Our benchmarks use the CombBLAS betweenness centrality code and benchmark included in the CombBLAS v1.5.0 distribution.
on the Blue Waters Cray XE6 supercomputer.
We provide a mix of performance benchmark tests, working with different types of graphs:
\begin{enumerate}
\item real-world social-network and citation graphs~\cite{snapnets},
\item synthetic weighted and unweighted R-MAT graphs~\cite{chakrabarti2004r},
\item Erd\H{o}s-R\'{e}nyi random graphs~\cite{gilbert1959random}.
\end{enumerate}
Real-world graphs and R-MAT graphs serve to provide effective strong scaling experiments.
We consider different types of weak scaling experiments using uniform graphs.
For {\it edge-weak scaling}, we keep the number of edges per processor and the nonzero fraction constant.
For {\it vertex-weak scaling}, we keep the number of vertices per processor and the average degree constant.

CTF-MFBC achieves consistent scalability patterns for all types of graphs.
It outperforms CombBLAS on various tests, but not uniformly so.
Large speed-ups are achieved for real, R-MAT, and uniform random graphs with the average vertex degree above 100.

\subsection{Experimental Setup}
To debug and tune our code, we used the NERSC Edison supercomputer, a Cray XC30 as well as the CSCS Piz Dora machine, a Cray XC40.
Each Edison compute node has two 12-core HT-enabled Intel Ivy Bridge sockets
with 64 GiB DDR3-1866 RAM.
Each node of Piz Dora has two 18-core Intel Broadwell CPUs (Intel® Xeon® E5-2695 v4). 
The network is the same on both of these machines, a Cray's Aries implementation of the Dragonfly
topology~\cite{dally08}. 
We then collected our final set of benchmarks on Blue Waters, a Cray XE6 supercomputer.
Each Blue Waters XE node has two 16-core AMD 6276 Interlagos sockets; there are 22,500 nodes in total. 
The network is a Cray Gemini torus.
The performance-portability of CTF made it easy to transition the code between these machines.
Our choice of machines was based on resource availability.

For both CTF and CombBLAS, we benchmarked a range of batch-sizes for each graph and processor count.
We show the highest performance rate over all batch sizes, which was usually achieved by the largest
batch-size that still fit in memory.
Such batches run for on the order of minutes, so we executed each batch
only once, rather than testing many iterations.
Due to the large overall time-granularity of the benchmarks,
we observed that system noise did not effect the runtime in the first 2 significant digits.


We use the metric of edge traversals per second (TEPS) to quantify performance.
The number of edge traversals scales with the size of the graph.
For betweenness centrality on a connected unweighted 
%
%
graph, each edge is traversed to consider shortest paths from every starting node.
We use all one MPI process per node and benchmark on core counts that are powers of four,
as CombBLAS requires square processor grids.
CTF can leverage threading and execute efficiently on most core-counts, but we maintain
powers of four for all experiments for consistency and simplicity.

\macb{Considered graphs}
We used two classes of synthetic graphs: R-MAT (power-law) and random-uniform.
We varied the density for both types of synthetic graphs.
Generally these graphs have a low diameter, which roughly reflects social-network graphs, 
a key application domain of betweenness centrality.
We also use real-world SNAP~\cite{snapnets} graphs (Table~\ref{tab:graphs}) of various sparsities and diameters.
%
%
Our CTF-MFBC code preprocessed all graphs to remove completely disconnected vertices.


\begin{table}[h]
\centering
\footnotesize
\sf
\begin{tabular}{@{}l|llllll@{}}
\toprule
\multicolumn{1}{c}{\textbf{ID}} & \multicolumn{1}{c}{\textbf{Name}} & directed? & \multicolumn{1}{c}{$n$} & \multicolumn{1}{c}{$m$} & \multicolumn{1}{c}{$d$}  & \multicolumn{1}{c}{$\bar{d}$} \\ \midrule
frd & Friendster              & undirected & 65.6M & 1.8B  & 32 & 5.8 \\
ork & Orkut social network		& undirected & 3.1M  & 117M  & 9  & 4.8 \\
ljm & LiveJournal membership	& directed   & 4.8M  & 70M   & 16 & 6.5 \\
cit & Patent citation graph		& directed   & 3.8M  & 16.5M & 22 & 9.4 \\
\bottomrule
\end{tabular}
\caption{The analyzed real-world graphs (sorted by ${m}$). All graphs are unweighted and have diameter $d$, with 90-percentile effective diameter $\bar{d}$~\cite{snapnets}.}
\label{tab:graphs}
\end{table}

\begin{figure*}[t]
\centering
\subfigure[Constant $n^2/p=n_0^2$ and edge percentage $f=100\cdot  m/n^2$]{
\includegraphics[width=3.3in]{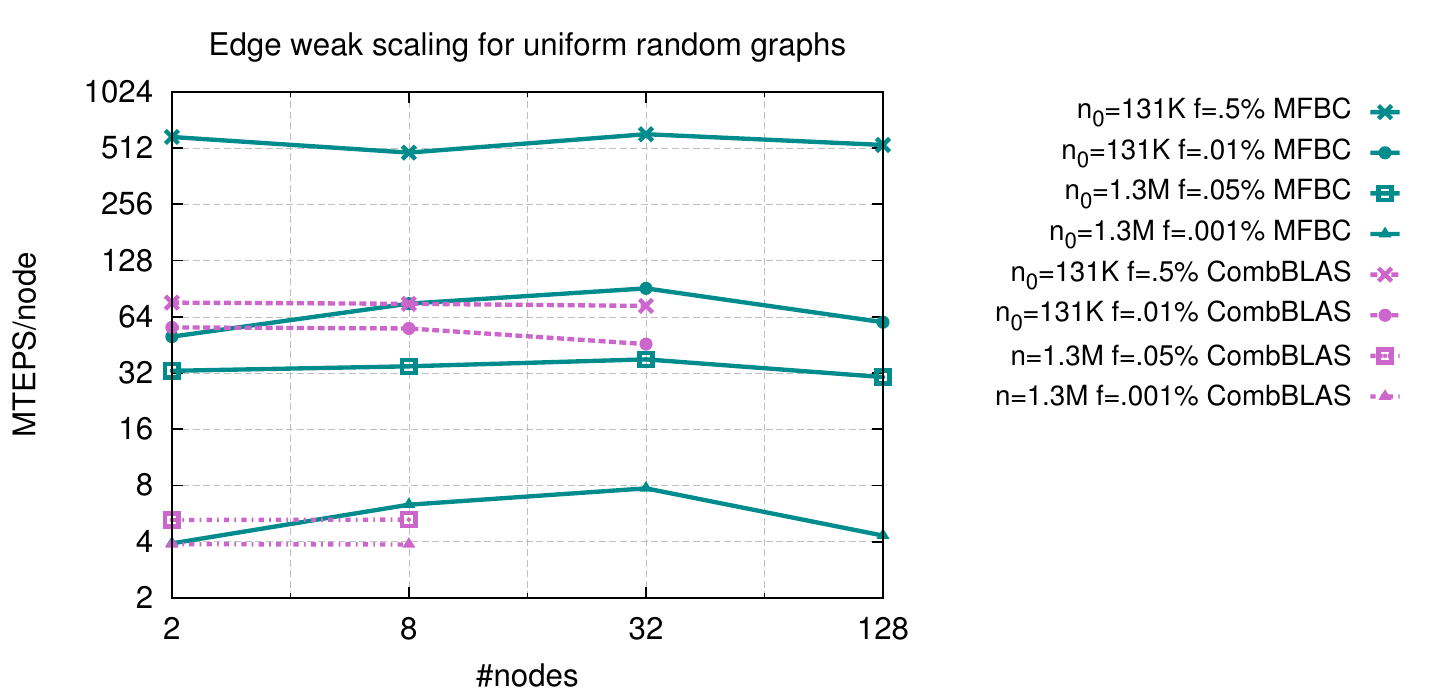}
\label{fig:ws_mf_vs_cb}
}
\subfigure[Constant $n/p=n_0$ and vertex degree $k=m/n$]{
\includegraphics[width=3.3in]{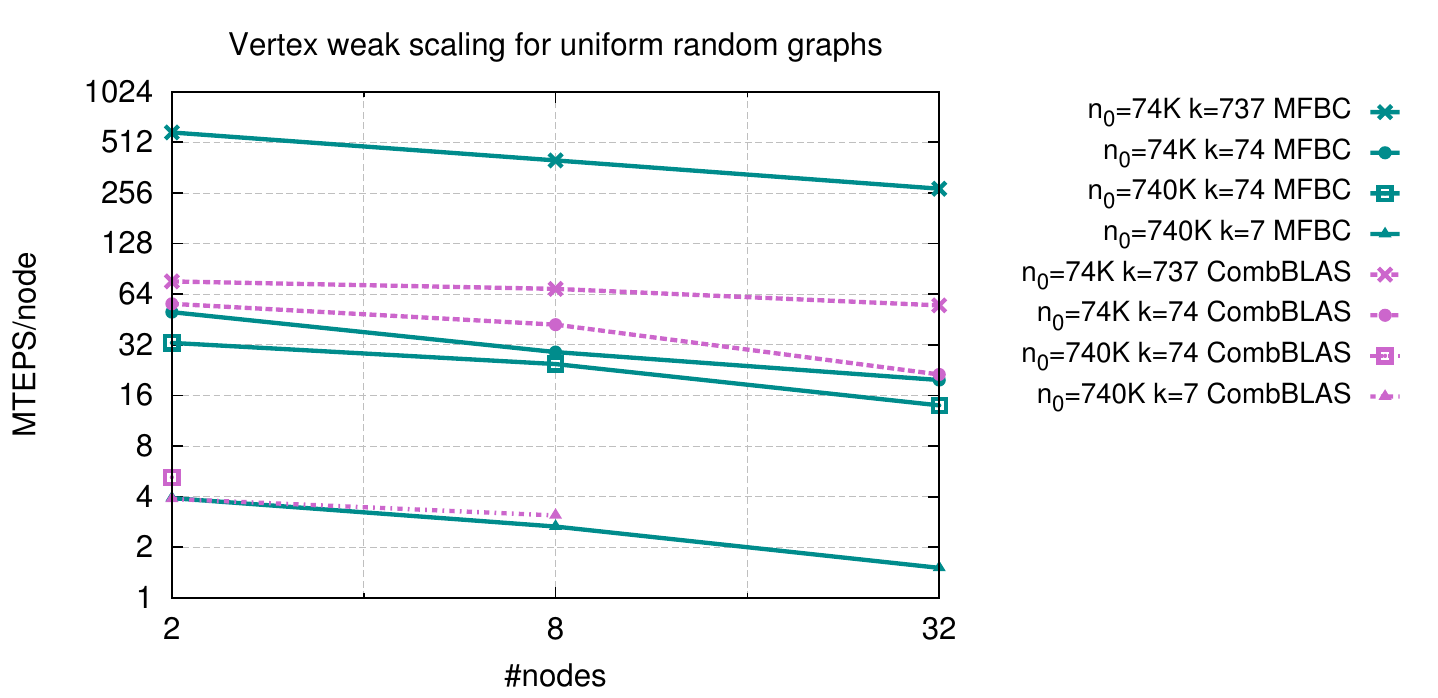}
\label{fig:ws2_mf_vs_cb}
}
\vspace{-1em}
\caption{Weak scaling of MFBC and CombBLAS for uniform random graphs}
\vspace{-1em}
\end{figure*}
\vspace{-3em}
\subsection{Strong Scaling}

We begin our performance study by testing the ability of MFBC to lower time to solution by using extra nodes (strong scaling).
Our strong scaling experiments evaluate MFBC with respect to CombBLAS on both real-world and synthetic graphs.

\macb{Real-world graphs}
Our selection of graphs (Table~\ref{tab:graphs}) considers three social-network graphs and a patent citation graph.
Two of the graphs are directed and two are undirected.
The diameter of the Orkut and LiveJournal graphs is relatively small, but Orkut is significantly denser.
The patent citation graph has the largest diameter, posing the biggest challenge to betweenness centrality computation.
The Friendster graph has roughly 15X more vertices and edges than any of the other graphs, and a fairly large diameter.

CTF-MFBC performs best for the Orkut graph, as the sparse matrix multiplications are more computation-intensive for denser graphs with low diameter (the latter implying denser frontiers).
The somewhat larger diameter of LiveJournal and the significantly larger diameter of the patent citation graph take a toll on the absolute performance.
However, the strong scalability is reasonable for all three of these graphs, as speed-ups of around 30X are achieved for each, when using 64X more nodes.
The smallest number of nodes on which CTF-MFBC successfully executes the Friendster graph is 32, with good scalability to 128 nodes.
The graph is the largest social network available in the SNAP dataset.

On the other hand, we observed relatively volatile performance for these graphs for CombBLAS, shown in Figure~\ref{fig:combblas_snap}.
In absolute terms the performance of CombBLAS on LiveJournal and the patent citation graph compare well to CTF-MFBC.
On the other hand, we were unable to successfully execute the Friendster graph with CombBLAS.
Further, CTF-MFBC is up to 7.6X faster for the Orkut graph.
We don't have a good understanding of the radical change in performance of the CombBLAS benchmark (in different directions for different graphs) when transitioning from 8 to 32 nodes.
On 32 nodes, CombBLAS reports a 10X improvement performance for LiveJournal by using a batch size of 8192 rather than 512, while for Orkut the corresponding improvement is only 3X.

\macb{R-MAT graphs}
We work with two R-MAT graphs, for both of which $\log_2(n) \approx S =22$, while the average degree is controlled by $k\approx E \in \{8,128\}$.
Disconnected vertices are removed in CTF-MFBC and skipped for consideration as starting vertices in CombBLAS.
R-MAT graphs have a low diameter, so a small number of matrix products is done in the unweighted case.

Figure~\ref{fig:combblas_ctf_rmat} compares CTF-MFBC with CombBLAS for strong scaling on these R-MAT graphs.
We omit points where we were not able to get CombBLAS to successfully execute betweenness centrality on the graph.
Overall, the performance is roughly the same for the case with small average degree E=8.
However, CTF-MFBC performs significantly better when E=128.

Figure~\ref{fig:combblas_ctf_rmat} also compares CTF-MFBC performance for R-MAT graphs with edge weights
randomly selected as integers in the range $[1,100]$ versus unweighted R-MAT.
In these tests, the number of sparse matrix multiplications doubles and the frontier stays relatively
dense for several steps of Algorithm~\ref{alg:MFBF}, thus the overall performance of MFBC decreases by more than a factor of two with the inclusion of weights.

\vspace{-0.5em}
\subsection{Weak Scaling}

We now test CTF-MFBC's parallel scalability, while keeping $m/p$ constant (weak scaling).
We use uniform random graphs, in which all nodes have the same
expected vertex degree, and every edge exists with a uniform probability.
We consider ``edge weak scaling'' where $n^2/p$ is kept constant and
``vertex weak scaling'' where $n/p$ is kept constant.
CTF-MFBC achieves good edge weak scaling, but deteriorates in efficiency for vertex weak scaling,
a discrepancy justified by our theoretical analysis.

Figure~\ref{fig:ws_mf_vs_cb} provides ``edge weak scaling'' results, in which the
sparsity percentage of the adjacency matrix, $f=100\cdot m/n^2$, stays constant.
The data confirms the observation that MFBC performs best for denser graphs.
CTF-MFBC scales well, which is expected, since the
communication cost term $O(\beta\cdot n^2/\sqrt{cp})$ grows in proportion with $\sqrt{p}$,
while the amount of computation per node $O(mn/p)$ also grows in proportion with $\sqrt{p}$.

Figure~\ref{fig:ws2_mf_vs_cb} provides ``vertex weak scaling'' results, in which the
the vertex degree $k$ stays constant.
We were unable to get CombBLAS to execute successfully on 64 nodes for the graphs with $n=740K$ vertices.
CTF-MFBC performs again better than CombBLAS when the average degree of the graph is large, but
both implementations deteriorate in performance rate with increasing node count.
This deterioration is predicted by our communication cost analysis, since in this weak scaling mode, the
term $O(\beta\cdot n^2/\sqrt{cp})$ grows in proportion with $p^{3/2}$, while the amount of work
per node $O(mn/p)$ grows in proportion with $p$.
Therefore, unlike edge weak scaling, vertex weak scaling is not sustainable, the number of words
communicated per unit of work grows with $\sqrt{p}$.

\subsection{Communication Cost}

We experimentally measured the communication complexity incurred by CTF-MFBC and CombBLAS.
Communication in both codes is dominated by collective communication routines.
We profiled  the time spent in communication (excluding time to synchronize) as well as used the parameters passed to MPI routines to build an analytical model.
To get the critical path costs, we follow the communication pattern: for each collective over a set of processors, we maximize the critical path costs incurred by those processors so far.
Broadcast and reduce of a message of size $n$ over $p$ processors have the cost $2n\cdot \beta+2\log_2(p)\cdot \alpha$ (twice that of scatter and allgather).
At the end of execution, we consider the maximum over all processors for each cost, thus obtaining the greatest amount of data communicated along any dependent sequence of collectives, as well
as the greatest amount of messages communicated along any (possibly different) dependent sequence of collectives.

\begin{table}[h]
{\small
\begin{tabular}{r|r|c|c|c|c|c}
graph & code & $W$ (GB) & $S$ (\#msgs) & comm (sec) & total (sec) \\
\hline
Orkut & CombBLAS & 19.71 & 546.9K & 48.81 & 233.7 \\
Orkut    & CTF-MFBC & 7.010 & 184.6K & 46.95 & 111.6 \\
LiveJournal & CombBLAS & 10.21 & 190.6K & 54.03 & 238.5 \\
LiveJournal & CTF-MFBC & 8.794 & 94.69K & 54.00 & 100.2 \\
Patents & CombBLAS & 1.026 & 202.8K & 4.084 & 6.422 \\
Patents & CTF-MFBC & 3.900 & 35.32K & 24.05 & 60.53 
\end{tabular}
}
\caption{Critical path times and costs collected on 4096 cores of Blue Waters, all for a single batch of 512 starting vertices. Total time includes overhead of additional profiling barriers.}
\label{tab:comm_res}
\end{table}
\vspace{-.2in}
Table~\ref{tab:comm_res} shows that CTF-MFBC uses fewer messages and performs less collective communication than CombBLAS in some cases.
In various other cases, CombBLAS performed less communication.
For the patent citations graph, CombBLAS performs significantly faster than CTF.
The blocked layout and choice of starting vertices likely permits CombBLAS to exploit locality for this directed graph.
Further, the back-propagation stage for CTF-MFBC (which is dynamically computed, rather than stored from BFS) takes considerably longer (performing more work in the directed case).

Overall, we can conclude that either the CombBLAS or CTF-MFBC implementation may incur more communication, depending on the type of graph.
CTF-MFBC seems to require less for denser graphs (e.g. Orkut and uniform random).
Persistence of layout and more accurate performance models would further reduce communication costs.

\section{Conclusion}
\label{sec:conc}

Our new maximal frontier algorithm for betweenness centrality achieves
good parallel scaling due to its low theoretical communication complexity
and a robust implementation of its primitive operations.
The algebraic formalism we use for propagating information through graphs
enables intuitive expression of frontiers and edge relaxations, making
it extensible to other graph problems such as maximum flow.
We expect that the approach of selecting frontiers to maximize overall progress
also leads to good parallel algorithms for other graph computations.

By implementing MFBC on top of Cyclops Tensor Framework (CTF), we have introduced the first application
case-study of CTF for a non-numerical problem.
MFBC with CTF shows substantial improvements in performance over CombBLAS for some graphs, while additionally being general to weighted graphs.
Automatic parallelism for sparse tensor contractions with arbitrary algebraic structures
is useful in many other application contexts.
The communication-efficiency achieved by sparse matrix multiplication routines in CombBLAS and CTF has 
a promising potential for changing the way massively-parallel graph computations are done.

\section{Acknowledgements}
This research is part of the Blue Waters sustained-petascale computing project, which is supported by the National Science Foundation (awards OCI-0725070 and ACI-1238993) and the state of Illinois. Blue Waters is a joint effort of the University of Illinois at Urbana-Champaign and its National Center for Supercomputing Applications.
We thank Hussein Harake, Colin McMurtrie, and the whole CSCS team granting access to the Piz Dora machine and for their excellent technical support. 
Maciej Besta was supported by Google European Doctoral Fellowship
and Edgar Solomonik was supported by an ETH Zurich Postdoctoral Fellowship.

\bibliographystyle{ACM-Reference-Format}
\bibliography{paper}

\end{document}